\documentclass[12pt,a4paper,english]{amsart}
\usepackage[utf8]{inputenc}
\usepackage[english]{babel} 

\usepackage{amsmath} 
\usepackage{amsthm} 
\usepackage{graphicx} 
\usepackage{xcolor} 
\usepackage{amsfonts}
\usepackage[biblabel]{cite}
\usepackage{bm} 
\usepackage[hidelinks]{hyperref} 
\usepackage{amssymb} 
\usepackage{tikz-cd} 
\usepackage{amscd}
\usepackage{etoolbox}
\patchcmd{\thmhead}{(#3)}{#3}{}{}
\usepackage{enumitem}
\usepackage{breqn} 

\DeclareMathOperator{\supp}{supp} 
\DeclareMathOperator{\ini}{in} 
\DeclareMathOperator{\N}{N} 
\DeclareMathOperator{\ev}{ev} 
\DeclareMathOperator{\Tr}{Tr}

\DeclareMathOperator{\wt}{wt}
\DeclareMathOperator{\Cart}{Cart}

\newcommand{\F}{{\mathbb{F}}}
\newcommand{\fq}{\mathbb{F}_q}
\newcommand{\fqs}{\mathbb{F}_{q^s}}

\newcommand{\xu}{\mathcal{X}_u}
\newcommand{\M}{\mathcal{M}}

\newcommand{\Md}{\mathcal{M}_{\leq d}}

\usepackage{mathtools} 
\DeclarePairedDelimiter\abs{\lvert}{\rvert}%
\DeclarePairedDelimiter\norm{\lVert}{\rVert}%

\makeatletter
\let\oldabs\abs
\def\abs{\@ifstar{\oldabs}{\oldabs*}}
\let\oldnorm\norm
\def\norm{\@ifstar{\oldnorm}{\oldnorm*}}
\makeatother

\newtheorem{thm}{Theorem}[section]
\newtheorem{prop}[thm]{Proposition}
\newtheorem{cor}[thm]{Corollary}
\newtheorem{lem}[thm]{Lemma}
\theoremstyle{definition}
\newtheorem{defn}[thm]{Definition} 

\newtheorem{rem}[thm]{Remark} 
 
\newtheorem{ex}[thm]{Example}

\def\fq{\mathbb{F}_q}
\def\fqs{\mathbb{F}_{q^s}}
\newcommand{\cal}[1]{\mathcal{#1}}
\newcommand{\cX}{\cal X}

\title[The weight hierarchy of decreasing norm-trace codes]{The weight hierarchy of decreasing norm-trace codes}

\author{Eduardo Camps-Moreno}
\address[Eduardo Camps-Moreno]{Department of Mathematics\\ Virginia Tech\\ 
Blacksburg, VA USA}
\email{eduardoc@vt.edu}

\author{Hiram H. L\'opez}
\address[Hiram H. L\'opez]{Department of Mathematics\\ Virginia Tech\\ 
Blacksburg, VA USA}
\email{hhlopez@vt.edu}

\author{Gretchen L. Matthews}
\address[Gretchen L. Matthews]{Department of Mathematics\\ Virginia Tech\\ Blacksburg, VA USA}
\email{gmatthews@vt.edu}

 \author{Rodrigo San-José}
 \address[Rodrigo San-José]{IMUVA-Mathematics Research Institute\\ Universidad de Valladolid\\ 47011 Valladolid, Spain}
\email{rodrigo.san-jose@uva.es}

\thanks{Hiram H. L\'opez was partially supported by NSF DMS-2401558. Gretchen L. Matthews was partially 
supported by NSF DMS-2201075 and the Commonwealth Cyber Initiative. Rodrigo San-José was partially supported by Grant TED2021-130358B-I00 funded by MICIU/AEI/10.13039/501100011033 and by the ``European Union NextGenerationEU/PRTR'', Grant PID2022-138906NB-C21 funded by MICIU/AEI/10.13039/501100011033 and by ERDF/EU, Grant QCAYLE supported by the European Union.-Next Generation UE/MICIU/PRTR/JCyL, and Grants FPU20/01311 and EST23/00777 funded by the Spanish Ministry of Universities.}



\begin{document}

\maketitle
\markleft{Eduardo Camps-Moreno, Hiram H. L\'opez, Gretchen L. Matthews, Rodrigo San-Jos\'e}

\begin{abstract}
The Generalized Hamming weights and their relative version, which generalize the minimum distance of a linear code, are relevant to numerous applications, including coding on the wire-tap channel of type II, $t$-resilient functions, bounding the cardinality of the output in list decoding algorithms, ramp secret sharing schemes, and quantum error correction. The generalized Hamming weights have been determined for some families of codes, including Cartesian codes and Hermitian one-point codes. In this paper, we determine the generalized Hamming weights of decreasing norm-trace codes, which are linear codes defined by evaluating monomials that are closed under divisibility on the rational points of the extended norm-trace curve given by $x^{u} = y^{q^{s - 1}} + y^{q^{s - 2}} + \cdots + y$ over the finite field of cardinality $q^s$, where $u$ is a positive divisor of $\frac{q^s - 1}{q - 1}$. As a particular case, we obtain the weight hierarchy of one-point norm-trace codes and recover the result of Barbero and Munuera (2001) giving the weight hierarchy of one-point Hermitian codes. We also study the relative generalized Hamming weights for these codes and use them to construct impure quantum codes with excellent parameters.
\end{abstract}

\section{Introduction}

Generalized Hamming weights for linear codes were introduced by Wei in 1991 \cite{weiGHW} as a tool in characterizing code performance on the wire-tap
channel of type II. Wei noted an application to $t$-resilient functions; established properties of the weight hierarchy, meaning the collection of generalized Hamming weights; and determined the generalized Hamming weights for first-order Reed-Muller codes, Hamming codes, extended Hamming codes, duals of Hamming codes, Reed-Solomon codes, and any maximum distance separable (MDS) code. Other related independent developments were Helleseth, Kl{\o}ve, and Mykkeltveit \cite{HELLESETH1977179_77, Helleseth_Klove_Ytrehus_92}. 
A flurry of activity followed. For instance, Feng, Tzeng, and Wei \cite{Feng_Tzeng_Wei_91} obtained bounds on the generalized Hamming weights of several families of BCH codes and cyclic codes; see also \cite{Janwa_Lal_97}. The generalized Hamming weights of convolutional codes were considered by Rosenthal and York in~\cite{Rosenthal_York_97}, and more recently with an alternative definition by Gorla and Salizzoni \cite{Gorla_Salizzoni_23}.

In~\cite{Yang_Kumar_Stichtenoth_94}, Yang, Kumar, and Stichtenoth initiated the study of generalized Hamming weights of algebraic geometry codes in 1994. In 1998, Heijnen and Pellikaan gave in~\cite{pellikaanGHWRM} an order bound on the generalized Hamming weights similar to the Feng-Rao bound on the minimum distance of a code. This order bound allowed Heijnen and Pellikaan to find the weight hierarchy for higher-order Reed-Muller codes and to provide insight into the family of one-point algebraic geometry codes. Later, Barbero and Munuera determined the weight hierarchy for one-point Hermitian codes in~\cite{Barbero_Munuera_00}. Additional work on generalized Hamming weights of algebraic geometry includes \cite{Delgado_13, BrasAmoros_14}.

Additional applications of generalized Hamming weights include improved bounds on list size in list decoding algorithms for tensor products and interleaved codes \cite{Gopalan_Guruswami_Raghavendra_11}. There has also been recent work on the generalized Hamming weights of many other well-known families of codes, such as affine Cartesian codes \cite{BEELEN2018130}, hyperbolic codes \cite{campsGHWHyperbolic}, projective Reed-Muller codes \cite{beelenGHWPRM,sanjoseRecursivePRM}, and matrix-product codes \cite{sanjoseGHWMPC}, to mention some.

In this paper, we determine the generalized Hamming weights of extended norm-trace codes. The weight hierarchies of the norm-trace and Hermitian codes are special cases of our results, allowing us to recover the results of Barbero and Munuera. The codes we consider are defined on the {\it extended norm-trace curve}, denoted by $\cX_u$, which is the affine curve defined over the finite field $\F_{q^s}$ with cardinality $q^s$  by the equation
\[
x^{u} = y^{q^{s - 1}} + y^{q^{s - 2}} + \cdots + y,
\]
where $u$ is a positive integer such that $u \mid \frac{q^s - 1}{q - 1}$.
Norm-trace codes are defined using the {\it norm-trace curve} which is the special case when $u= \frac{q^s - 1}{q - 1}$, meaning the  affine curve over $\F_{q^s}$ defined by \[ \N(x)=\Tr(y) \]
where $\N(x)$ is the norm and $\Tr(y)$ is the trace, both taken with respect to the 
extension $\F_{q^s}/\F_q$. Taking in addition $s=2$ and $u= \frac{q^s - 1}{q - 1}$ (resp., $u \mid \frac{q^s - 1}{q - 1}$) gives the special case of the Hermitian curve defined by $x^{q+1}=y^q+y$ (resp., its quotient $x^{u}=y^q+y$) over $\F_{q^2}$.

We focus in this work on decreasing norm-trace codes, which are codes defined by evaluating monomials on the rational points of the curve $\cX_u$. To define the codes, enumerate the rational points on $\cX_u$ so that $\cX_u = \left\{P_1,\ldots,P_n \right\} \subseteq \fqs^2$ and $n=u(q-1)q^{s-1}+q^{s-1}$.
The \textit{evaluation map}, denoted ${\rm ev}$, is the $\fqs$-linear map given by  
$$
\begin{array}{lccc}
{\rm ev}\colon &\fqs[x,y] &\rightarrow& \fqs^{n}\quad \\
&f & \mapsto& \left(f(P_1),\ldots,f(P_n)\right).
\end{array}
$$
Let $\mathcal{M} \subseteq \fqs[x,y]$ be a set of monomials closed under 
divisibility, meaning that if $M\in \mathcal{M}$ and $M^\prime$ divides $M$, 
then $M^\prime \in \mathcal{M}$. Let $\mathcal{L}(\M)$ be the 
$\fqs$-subspace of $\fqs[x,y]$ generated by the set $\mathcal{M}$. We call the image of 
$\mathcal{L}(\M)$ under the evaluation map, denoted by ${\rm ev}(\mathcal{M})$, a {\it decreasing norm-trace code}. 
We can see that the extended norm-trace codes introduced and recently studied 
in \cite{bras-amor} and \cite{Heera-Pin} are particular instances of decreasing 
norm-trace codes; norm-trace codes and generalizations also appear in \cite{BO}. Moreover, we check later that the family of decreasing 
norm-trace codes contains, as a specific case, the family of one-point 
algebraic geometry codes over the norm-trace curve (in particular, it contains the family of one-point algebraic geometry codes over the Hermitian curve).

We organize this paper as follows. The necessary background is contained in Section \ref{sec:prelim}. In Section \ref{sec:GHW_dec}, we determine the weight hierarchy of decreasing norm-trace codes. 
In Section \ref{sec:GHW_RMlike}, we find the generalized Hamming weights for a class of Reed-Muller-type codes defined over the norm-trace curve. In Section \ref{sec:RGHW}, we show how to adapt our techniques to study the relative generalized Hamming weights of decreasing norm-trace codes. As an application, we consider one-point algebraic geometry codes on the extended norm-trace curve, a particular case of decreasing norm-trace codes, and we use them to construct impure quantum codes. 

\section{Preliminaries} \label{sec:prelim}

We use the standard notation from finite fields and coding theory. Let $\fq$ be the finite field with $q$ elements. The set of polynomials in variables $x_1, \dots x_m$ with coefficients in $\F_q$ is denoted by $\F_q[x_1, \dots, x_m]$. The affine plane over $\F_q$ is denoted by $\mathbb A^2(\F_q)$.

Let $C$ be an $[n,k,d]$ code $C$ over $\F_q$, meaning a $k$-dimensional $\fq$-subspace of $\F_q^n$ with minimum distance $d:=\left\{ \wt(c): c \in C \setminus \{ 0\} \right\}$ where the weight of $w \in \F_q^n$ is $\wt(w):=\abs{\left\{ i \in [n]: w_i \neq 0 \right\}}$ and $[n]:=\left\{ 1, \dots, n \right\}$. The support of the  code $C$ is 
$$
\supp(C):=\left\{ i \in [n]: \exists c \in C \textnormal{ with } c_i \neq 0 \right\}.
$$

\begin{defn} \label{def:ghw}
       The $r^{th}$-generalized Hamming weight of an $[n,k,d]$ code $C$ is 
$$
d_r(C):=\min \left\{ \abs{ \supp(C') }: C' \textnormal{ is a subcode of } C \textnormal{ of dimension } r \right\},
$$
where $r \in [k]$. The weight hierarchy of $C$ is the set 
$$
\left\{ d_r(C): r \in [k] \right\}. 
$$
\end{defn}

Notice that $d_1(C)$ is the minimum distance of $C$. We may write $d_r$ to mean $d_r(C)$ if the code $C$ is evident from the context. Wei obtained the following general properties of the generalized Hamming weights in \cite{weiGHW}.

\begin{thm}[(Monotonicity)]\label{T:monotonicity}
For an $[n,k]$ linear code $C$ with $k>0$, we have
$$
1\leq d_1(C)<d_2(C)<\cdots <d_k(C)\leq n.
$$
\end{thm}
\begin{cor}[(Generalized Singleton Bound)]\label{C:singletongeneralizada}
For an $[n,k]$ linear code $C$, we have
$$
d_r(C)\leq n-k+r, \; 1\leq r\leq k.
$$
\end{cor}
\begin{thm}[(Duality)]\label{ghwdual}
Let $C$ be an $[n,k]$ code. Then
$$
\{d_r(C):1\leq r\leq k\}=\{1,2,\dots,n\}\setminus \{n+1-d_r(C^\perp):1\leq r\leq n-k\}.
$$
\end{thm}

We will now introduce some basic tools from commutative algebra and Gröbner basis that we will use throughout the article. We refer to \cite{cox, Villarreal2018-kz} for the details we may not cover in what follows.

Let $\M$ be a set  of monomials of $\F_{q^s}[x_1, \ldots, x_m]$ with a monomial order $\prec$, meaning a total order in which the least monomial is $1$ and $M_1 \prec M_2$ implies $M M_1 \prec M M_2$ for all $M, M_1, M_2 \in \M$. For  $f \in \F_{q^s}[x_1,\ldots, x_m] \setminus \{ 0 \}$, the {\em initial (or leading) monomial} of $f$ is  
the greatest monomial with respect to $\prec$ which appears in the expansion of $f$, denoted by $\ini(f)$. Given an ideal $I \subseteq \F_{q^s}[x_1, \ldots, x_m]$, its {\em initial ideal} is $\ini(I):= (\ini(f): f \in I )$. A {\em Gr\"{o}bner basis} for $I$ is a finite set $\mathcal G
\subseteq I$ such that for every polynomial $f \in I \setminus \{0\}$, $\ini(f)$ is a multiple of $\ini(g)$ for some $g \in \mathcal G$. In this case, $\mathcal{G}$ generates $I$. Given a Gr\"{o}bner basis $\mathcal G$ for $I$, the {\em footprint} of $I$ is the set 
$$
\Delta(I):=\left\{ M \in \M: M \not \in \ini(I) \right\}=\left\{ M \in \M: \ini(g) \nmid M \ \ \forall  g \in \mathcal G \right\}.
$$ 
Let $A\subset \fqs[x_1,\dots,x_m]$. We may abuse the notation and define $\Delta(A)$ as the footprint of the ideal generated by $A$. For an ideal $I\subset \F_{q^s}[x_1,\ldots, x_m]$, the set $\Delta(I)$ forms a basis for $\F_{q^s}[x_1,\ldots, x_m]/I$ as an $\fqs$-vector space. Therefore, if we consider a set of (classes of) polynomials in $\F_{q^s}[x_1,\ldots, x_m]/I$, we may assume that they only have monomials from $\Delta(I)$ in their expansion, and if they differ in at least one monomial (for example, the initial monomial), then they are linearly independent in $\F_{q^s}[x_1,\ldots, x_m]/I$.

The $S$-polynomial of $g,g' \in \mathcal G$, with leading coefficients $c,c'$, respectively, is the polynomial
$$
S(g,g')=c'f'g-cfg' \in I,
$$ 
where $f, f' \in \F_{q^s}[x_1, \ldots, x_m]$ are the monomials of least degree such that $ f' \ini(g)=f \ini(g')$. For an ideal $I$ whose associated variety $V(I)$ is finite, the {\it footprint bound} states that the variety has cardinality bounded by the footprint of the ideal, that is
$$
\abs{V(I)} \leq \abs{\Delta(I)}.
$$
Moreover, equality holds when $I$ is a radical ideal (for the details, see \cite[Thm. 6 and Prop. 7, Chapter 5 \textsection 3]{cox}). 

We now review the necessary prerequisites on the curves of interest. Let $s \geq 2$ be an integer. Define
the polynomials $\N(x):=x^{\frac{q^s-1}{q-1}}$ and $\Tr(y):=y^{q^{s - 1}} + 
y^{q^{s - 2}} + \cdots + y^q + y$ in $\fqs[x,y]$. The {\it trace} with respect to the 
extension $\F_{q^s}/\F_q$  is  the map
$$
\begin{array}{lccc}
\Tr: & \F_{q^s} & \to & \F_q \\
& \alpha & \mapsto & \Tr(\alpha)=\alpha+\alpha^q+\cdots +\alpha^{q^{s-1}}.
\end{array}
$$
The {\it norm} with respect to the extension $\F_{q^s}/\F_q$  is  the 
map
$$
\begin{array}{lccc}
\N: & \F_{q^s} & \rightarrow & \F_q \\
& \alpha & \mapsto & \N(\alpha)=\alpha\cdot \alpha^q \cdots \alpha^{q^{s-1}}=\alpha^{\frac{q^s-1}{q-1}}.
\end{array}
$$
The {\it norm-trace curve}, denoted by $\mathcal{X}$, is the affine plane curve over $\fqs$ given by the equation 
\[
x^{\frac{q^s - 1}{q - 1}} = y^{q^{s - 1}} + y^{q^{s - 2}} + \cdots + y.
\]
The curve $\mathcal{X}$ has been extensively studied in the literature to construct linear codes~\cite{bras-amor, geil, kim_boran, nt_lifted, MTT_08}. We are interested in a slightly more general curve. Let $u$ be a positive integer such that $u \mid \frac{q^s - 1}{q - 1}$.
The {\it extended norm-trace curve}, denoted by $\cX_u$, is the affine curve over $\F_{q^s}$ defined by the equation
\[
x^{u} = y^{q^{s - 1}} + y^{q^{s - 2}} + \cdots + y.
\]
The curve $\mathcal X_u$ has genus $g=\frac{(u-1)(q^s-1)}{2}$~\cite[Theorem 13]{nt_lifted}. In~\cite{decreasingnormtrace}, the authors use the rational points of the curve $\cX_u$ to construct a family of decreasing evaluation codes, which contains, as a particular case, the extended norm-trace codes~\cite{bras-amor, Heera-Pin}. In particular, in~\cite{decreasingnormtrace}, the authors obtain the following result about the structure of the rational points of $\mathcal X_u$.

\begin{lem} \label{puntos}
For every $\gamma\in \fq$, define $A_\gamma:=\{(\alpha,\beta)\in \mathbb{A}^2(\fqs): \Tr(\beta)=\alpha^u=\gamma\}$. Then, we have $\xu=\bigcup_{\gamma\in\fq}A_\gamma$. Moreover, $\abs{A_0}=q^{s-1}$ and $\abs{A_\gamma}=uq^{s-1}$ for all $\gamma \in \fq^*$.
\end{lem}

\begin{rem}\label{rempuntos}
For each $\gamma \in \fq$, we have that
$$
A_\gamma=\{\alpha\in \fqs : \alpha^u=\gamma\} \times \{\beta\in \fqs : \Tr(\beta)=\gamma\}.
$$
Fix $\gamma\in \fq$. For each $\alpha$ with $\alpha^u=\gamma$, there are exactly $q^{s-1}$ points in $\xu$ such that its first coordinate is $\alpha$. For each $\beta$ such that $\Tr(\beta)=\gamma$, there are exactly $u$ points in $\xu$ such that its second coordinate is $\beta$ if $\gamma\neq 0$, and one point if $\gamma=0$. 
\end{rem}

The {\it vanishing ideal} of a set of points is defined by the set of all polynomials that vanish on every point. Consider the vanishing ideal $I_{\xu}:=I(\xu)=(\Tr(y)-x^u,x^{q^s}-x,y^{q^s}-y)$ associated with the extended norm-trace curve. From \cite[Prop. 2.2]{decreasingnormtrace}, we have 
\begin{equation}\label{eq:vanishing}
I_{\xu}=(\Tr(y)-x^u,x^{u(q-1)+1}-x).
\end{equation}
We define the weight $w(x^ay^b)$ of the monomial $x^ay^b$ as
$$
w(x^ay^b)=aq^{s-1}+ub.
$$
The weight of a polynomial is the maximum weight of the monomials that appears in the support of the polynomial. We consider the weighted degree lexicographic ordering given by the following:
$$
x^ay^b < x^{a'}y^{b'} \iff w(x^ay^b)<w(x^{a'}y^{b'}), \text{ or } w(x^ay^b)=w(x^{a'}y^{b'}) \text{ and } b<b'.
$$
For this weighted degree lexicographic order, in~\cite{decreasingnormtrace}, the authors prove that the generators from Equation (\ref{eq:vanishing}) form a Gröbner basis, and therefore
\begin{equation}\label{eq:initial}
\ini(I_{\xu})=(y^{q^{s-1}},x^{u(q-1)+1}).
\end{equation}
Given $f\in \fqs[x,y]$, we denote
$$
V_\xu(f):=\{ P\in \xu \mid f(P)=0\}.
$$
We have that 
$$
\abs{V_\xu(f)}=\abs{\Delta(I_\xu,f)}=\abs{\Delta(\ini(I_\xu,f))}.
$$
The second equality is always true. The first one follows from \cite[Prop 6.2.12]{kreuzer3}. Assume $\ini(f)=x^ay^b$, with $0\leq a\leq u(q-1)$ and $0\leq b\leq q^{s-1}-1$. Following the proof of \cite[Prop. 4]{geilbezout}, let $f=x^ay^b+f'$ and $g=\Tr(y)-x^u=y^{q^{s-1}}-x^u+g'$, with $w(f')<aq^{s-1}+bu$ and $w(g')<uq^{s-1}$. Then we can compute the following $S$-polynomial:
$$
S(g,f)=x^ag-y^{q^{s-1}-b}f=-x^{a+u}+x^ag'-y^{q^{s-1}-b}f'.
$$
We have $\ini(S(g,f))=x^{a+u}$ because $w(x^ag')<a+uq^{s-1}$ and $w(y^{q^{s-1}-b}f')<(a+u)q^{s-1}$, whereas $w(x^{a+u})=(a+u)q^{s-1}$. Therefore, we have $(\ini(I_\xu),x^ay^b,x^{a+u}) \subset\ini(I_\xu,f)$, which implies 
$$
\abs{V_\xu(f)}=\abs{\Delta(\ini(I_\xu,f))}\leq \abs{\Delta(\ini(I_\xu),x^ay^b,x^{a+u})}.
$$
Using Equation (\ref{eq:initial}), we can conclude that 
$$
\abs{V_\xu(f)}\leq\abs{ \Delta(y^{q^{s-1}},x^{\min\{a+u, u(q-1)+1\}},x^{a}y^{b})}. 
$$

Let $F=\{f_1,\dots,f_r\}$ be a set of linearly independent polynomials in $\fqs[x,y]/I_\xu$. Let $\ini(f_i)=x^{a_i}y^{b_i}$, for $1\leq i \leq r$, with $a_1\leq \cdots \leq a_r$. We can assume that $(a_i,b_i)\neq (a_j,b_j)$ for $i\neq j$. We define $$F_{\ini}:=\{ \ini(f): f\in F\}=\{ x^{a_i}y^{b_i}: 1\leq i \leq r\}.$$ Arguing as above, we can get the bound
\begin{equation}\label{lowerbound}
\abs{V_\xu(F)}\leq \abs{\Delta(\{y^{q^{s-1}},x^{\min\{a_1+u, u(q-1)+1\}}\} \cup F_\ini)},
\end{equation}
where $V_\xu(F)$ is the set of common zeroes of $F$ in $\xu$. Since this will be the main object of study, we denote
$$
\Delta^*(F):=\Delta(\{y^{q^{s-1}},x^{\min\{a_1+u, u(q-1)+1\}}\} \cup F_\ini).
$$
As the number of common zeroes of $F$ is related to the support of the subcode $D$ generated by the evaluation of the polynomials in $F$, Inequality~(\ref{lowerbound}) gives the following bound for the support of this subcode:
\begin{equation}\label{eq:boundsupp}
\abs{\supp(D)}\geq \abs{\xu}-\abs{\Delta^*(F)}.
\end{equation}

\section{Decreasing norm-trace codes} \label{sec:GHW_dec}

Let $\mathcal{L}(\M) \subseteq \fqs[x,y]$ be the $\fqs$-subspace generated by a set of monomials $\mathcal{M} \subseteq \fqs[x,y]$ that is closed under divisibility. Let $\left\{P_1,\ldots,P_n \right\} \subseteq \fqs^2$ be the rational points  of the extended norm-trace curve $\cX_u$ defined over the finite field $\F_{q^s}$ by the equation
\[
x^{u} = y^{q^{s - 1}} + y^{q^{s - 2}} + \cdots + y,
\]
where $n=u(q-1)q^{s-1}+q^{s-1}$ and $u$ is a positive integer such that $u \mid \frac{q^s - 1}{q - 1}$. Recall a decreasing norm-trace code is defined by
\[ {\rm ev}(\mathcal{M}) = \{ \left(f(P_1),\ldots,f(P_n)\right) : f \in \mathcal{L}(\M) \}.\]

In this section, we determine the weight hierarchy of decreasing norm-trace codes by proving that the footprint-like bound given in Inequality~(\ref{lowerbound}) is always attained with these codes. We start with a lemma that provides a formula for the number of elements in the footprint.

\begin{lem}\label{lemafootprint}
Let $r\geq 1$ and $Z=\{x^{a_i}y^{b_i}:1\leq i\leq r\}\subset \Delta(y^{q^{s-1}},x^{u(q-1)+1})$ be such that $a_1<\cdots <a_r$, $b_1>\cdots > b_r$. Additionally, assume that $a_r-a_1<u$. Let $v=\min\{a_1+u, u(q-1)+1\}$. Then 
$$
\abs{\Delta^*(Z)}=\abs{\Delta(\{y^{q^{s-1}},x^{v}\} \cup Z)}=a_1q^{s-1}+b_r(v-a_1)+\sum_{i=1}^{r-1} (a_{i+1}-a_i)(b_i-b_r).
$$
\end{lem}

\begin{proof}
As $a_r-a_1<u$, we have $a_i<v$ for $i=1,\dots,r$. We have the decomposition
$$
\Delta(\{y^{q^{s-1}},x^{\min\{a_1+u, u(q-1)+1\}}\} \cup Z)=\Delta(y^{q^{s-1}},x^v,x^{a_1}y^{b_r})\cup (x^{a_1}y^{b_r}\cdot \Delta(Z')),
$$
where $Z'=\{x^{a_i-a_1}y^{b_i-b_r},1\leq i\leq r\}$ (see Figure \ref{f:fp}). It is easy to obtain that
$$
\abs{\Delta(y^{q^{s-1}},x^v,x^{a_1}y^{b_r})}=a_1q^{s-1}+b_r(v-a_1).
$$
Finally, we have that
$$
\abs{\Delta(Z')}=\sum_{i=1}^{r-1} (a_{i+1}-a_i)(b_i-b_r).
$$
This completes the proof.
\end{proof}

\begin{figure}[hbt!]
\includegraphics[width=0.5\textwidth]{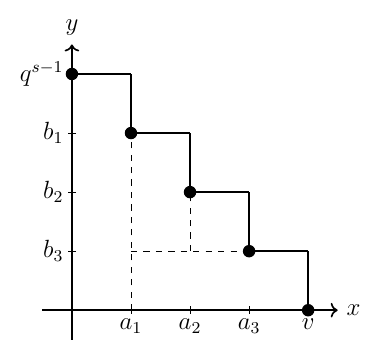}
\caption{Example of the decomposition of the footprint used in the proof of Lemma~\ref{lemafootprint} with $r=3$.}\label{f:fp}
\end{figure}

We present now the main result of this section, in which we show how to construct sets of polynomials attaining the footprint-like bound in Inequality~(\ref{lowerbound}).

\begin{thm}\label{T:fpsharp}
Let $\M\subset \Delta(y^{q^{s-1}},x^{u(q-1)+1})$ be a decreasing set and $1\leq r\leq \abs{\M}$.
Let $\mathcal{N}_r:=\{ N=\{M_1,\dots,M_r\}\subset\M: M_i\neq M_j \text{ for } i\neq j \}$.
Then, the $r^{th}$ generalized Hamming weight of the decreasing extended norm-trace code is 
$$
d_r(\ev(\M))=\abs{\xu}-\max\left\{\abs{\Delta^*(N)}, N \in \mathcal{N}_r\right\}. 
$$
\end{thm}
\begin{proof}
We consider $N=\{x^{a_i}y^{b_i}, 1\leq i\leq r\}\subset \M$, with $a_1\leq \cdots \leq a_r$. We will find a set of polynomials $F$ such that $N=F_\ini$ and
$$
\abs{V_\xu(F)}=\abs{\Delta^*(F)}=\abs{\Delta^*(F_\ini)}=\abs{\Delta^*(N)},
$$
which proves the result by also considering Inequality~(\ref{eq:boundsupp}). We now divide the proof into the following cases:
$$
\begin{cases}
\text{(1)} & a_1\leq u(q-2)+1
\begin{cases}
\text{(1.1)} & N \text{ as in Lemma~\ref{lemafootprint}}\\
\text{(1.2)} & N \text{ not as in Lemma~\ref{lemafootprint}}
\begin{cases}
\text{(1.2.1)} & x^{a_j}y^{b_j} | x^ay^b\\
\text{(1.2.2)} & x^{a_1+u} | x^ay^b\\
\end{cases}\\
\end{cases}\\
\text{(2)} & a_1\geq u(q-2)+2
\begin{cases} \text{(2.1)} & N \text{ as in Lemma~\ref{lemafootprint}}\\
\text{(2.2)} & N \text{ not as in Lemma~\ref{lemafootprint}}\end{cases}
\\
\end{cases}
$$
Case (1) Assume $a_1+u\leq u(q-1)+1$, i.e., $a_1\leq u(q-2)+1$.
\newline
Case (1.1) Consider the case when $N$ is as in Lemma~\ref{lemafootprint}. In particular, $a_r-a_1<u$. Fix $\gamma \in \fq^*$. In what follows, whenever we consider a collection of elements of $\fqs$, we may assume they are pairwise distinct. Consider elements $\alpha_1,\dots,\alpha_{a_1}$ such that $\alpha_i^u\neq \gamma$, for $1\leq i \leq a_1$. Note that we can do this because $a_1\leq u(q-2)+1$ and Remark \ref{rempuntos}. We also consider $\alpha'_1,\dots,\alpha'_{a_r-a_1}$ such that $(\alpha'_i)^u=\gamma$, for $1\leq i \leq a_r-a_1$ (recall $a_r-a_1<u$), and $\beta_1,\dots,\beta_{b_1}$ such that $\Tr(\beta_i)=\gamma$, for $1\leq i \leq b_1$. Take $F := \{f_1,\dots,f_r\}$, where
$$
f_j=\prod_{i=1}^{a_1}(x-\alpha_i)\prod_{i=1}^{a_j-a_1}(x-\alpha'_i)\prod_{i=1}^{b_j}(y-\beta_i)
$$
for $j=1,\dots,r$. By construction, any point $(\alpha,\beta)\in \xu$, with $\alpha=\alpha_i$, for some $1\leq i\leq a_1$, is a common zero of $F$. The same happens if $\beta=\beta_i$, for some $1\leq i\leq b_r$. These two sets of zeroes are disjoint since $\alpha_i^u\neq \gamma=\Tr(\beta_i)$. We obtain $a_1q^{s-1}+b_r u$ zeroes in this way. By Lemma \ref{puntos} and Remark \ref{rempuntos}, and by inspection of the polynomials in $F$, it is clear that any other common zero $(\alpha,\beta)$ of $F$ must have $\alpha=\alpha'_i$ for some $1\leq i \leq a_r-a_1$, and $\beta=\beta_i$ for some $b_r+1\leq i\leq b_1$. If $\alpha=\alpha'_i$, with $a_{j}-a_1+1\leq i\leq a_{j+1}-a_1$ for some $1\leq j\leq r-1$,
then $(\alpha,\beta)$ is a common zero of $\{f_{j+1},\dots,f_{r}\}$. As $(x-\alpha)$ is not a factor of the rest of the polynomials in $F$, we must have $\beta=\beta_\ell$ for some $b_{r+1}+1\leq \ell\leq b_j$. This gives $(a_{j+1}-a_j)(b_j-b_{r})$ zeroes, counting for all the possible $i$ with $a_j-a_1+1\leq i\leq a_{j+1}-a_1$ and $\ell$ with $b_{r+1}+1\leq \ell \leq b_j$. Considering the sum over all possible $j$, with $1\leq j \leq r-1$, and using Lemma \ref{lemafootprint}, we obtain that 
$$
\abs{V_\xu(F)}=\abs{\Delta^*(F)}=\abs{\Delta(\{y^{q^{s-1}},x^{\min\{a_1+u, u(q-1)+1\}}\} \cup F_\ini)}.
$$
Case (1.2) Assume that $N$ is not as in Lemma \ref{lemafootprint}. Then, when we consider the set
$$
\{y^{q^{s-1}},x^{a_1+u}\} \cup N,
$$
there are some monomials that are multiples of the others. Therefore, we can consider $N'\subset N$ such that $N'$ satisfies the conditions from Lemma \ref{lemafootprint}. Arguing as in Case~(1.1), we can find a set of $r'<r$ polynomials $F'$ such that $F'_\ini=N'$ and
$$
\abs{V_\xu(F')}=\abs{\Delta(\{y^{q^{s-1}},x^{a_1+u}\} \cup N')}=\abs{\Delta(\{y^{q^{s-1}},x^{a_1+u}\} \cup N)}.
$$
For each monomial $x^ay^b\in N\setminus N'$, we know that there is some $j$ such that $x^{a_j}y^{b_j}\in N'$ and $x^{a_j}y^{b_j}$ divides $x^ay^b$, or $x^{a_1+u}$ divides $x^{a}y^b$. 
\newline
Case (1.2.1) If $x^{a_j}y^{b_j}$ divides $x^ay^b$, we consider the polynomial
$$
f_{a,b}=f_j(x-\alpha_1)^{a-a_j}(y-\beta_1)^{b-b_j}.
$$
It is clear that 
\begin{equation}\label{eqceros}
\abs{V_\xu(F')}=\abs{V_\xu(F'\cup \{f_{a,b}\})}
\end{equation}
because we are just increasing the multiplicity of the zeroes of $f_j\in F'$, and $\ini(f_{a,b})=x^ay^b$. 
\newline
Case (1.2.2) If $x^{a_1+u}$ divides $x^ay^b$, we consider the polynomial
$$
f_{a,b}=(x^u-\gamma)\left(\prod_{i=1}^{a_1}(x-\alpha_i)\right)(x-\alpha_1)^{a-a_1-u}(y-\beta_1)^{b}.
$$
Here, Equation~(\ref{eqceros}) also holds because any common zero $(\alpha,\beta)$ of $F'$ has $\alpha=\alpha_i$ for some $1\leq i \leq a_1$, or $\alpha^u=\gamma$. We also have $\ini(f_{a,b})=x^ay^b$. 

In both cases (1.2.1) and (1.2.2), we obtain that, for any common zero $(\alpha,\beta)$ of $F'$, $(\alpha,\beta)$ is also a zero of $f_{a,b}$.
Therefore, applying this procedure for every $x^ay^b\in N\setminus N'$, we can obtain a set $F$ with
$$
\abs{V_\xu(F)}=\abs{V_\xu(F')}=\abs{\Delta^*(F)}
$$
and $F_\ini=N$.

\noindent Case (2) We assume $a_1\geq u(q-2)+2$. We fix $\gamma\in \fq^*$. As $\mathcal{M}\subset \Delta(y^{q^{s-1}},x^{u(q-1)+1})$, we have $a_i\leq u(q-1)$, for $1\leq i \leq r$, which implies $a_r-u(q-2)\leq u$. Hence, we can choose $\alpha_1,\dots,\alpha_{u(q-2)+1}$ such that $\alpha_i^u \neq \gamma$, for $1\leq i\leq u(q-2)+1$, and $\alpha'_1,\dots,\alpha'_{a_r-u(q-2)-1}$ such that $(\alpha'_i)^u=\gamma$, for $1\leq i \leq a_r-u(q-2)-1$. Note that, since $u(q-2)+2\leq a_1\leq a_r$, we have $a_r-u(q-2)-1\geq 1$. As before, we consider $\beta_1,\dots,\beta_{b_1}$ such that $\Tr(\beta_i)=\gamma$, for $1\leq i \leq b_1$.
\newline
Case (2.1) Assuming $N$ is as in Lemma \ref{lemafootprint}, we consider the polynomials
$$
f_j=\prod_{i=1}^{u(q-2)+1}(x-\alpha_i)\prod_{i=1}^{a_j-u(q-2)-1}(x-\alpha'_i)\prod_{i=1}^{b_j}(y-\beta_i),
$$
for $j=1,\dots,r$, and $F=\{f_1,\dots,f_r\}$. In order to count the number of common zeroes, we note that
$$
\prod_{i=1}^{u(q-2)+1}(x-\alpha_i)\prod_{i=1}^{a_1-u(q-2)-1}(x-\alpha'_i)\prod_{i=1}^{b_r}(y-\beta_i)
$$
divides $f_j$, for $1\leq j \leq r$. This polynomial has $a_1q^{s-1}$ zeroes $(\alpha,\beta)$ with $\alpha=\alpha_i$ or $\alpha=\alpha_i'$, and $ub_r$ zeroes with $\beta=\beta_i$, for some $i$. If we consider these zeroes together, we are considering the zeroes with $\alpha=\alpha'_i$ twice since $(\alpha'_i)^u=\gamma$, which are $(a_1-u(q-2)-1)b_r$ zeroes. Thus, in this way we obtain $a_1q^{s-1}+(u(q-1)+1-a_1)b_r$ common zeroes of $F$. Looking at the rest of the factors of the $f_j$ and counting as before, we obtain $a_1q^{s-1}+(u(q-1)+1-a_1)b_r+\sum_{i=1}^{r-1}(a_{i+1}-a_i)(b_i-b_r)$ common zeroes in total. By Lemma \ref{lemafootprint}, we obtain that $\abs{V_\xu(F)}=\abs{\Delta^*(F)}$.
\newline
Case (2.2) If $N$ is not as in Lemma \ref{lemafootprint}, we can use the argument we used above with a subset $N'\subset N$ that satisfies the conditions in Lemma~\ref{lemafootprint}, and then increase the multiplicity of the zeroes of some of the polynomials.
\end{proof}

\begin{rem}
In Theorem \ref{T:fpsharp}, requiring $\M$ to be a decreasing set ensures that all the polynomials considered in the proof belong to $\mathcal{L}(\M)$. For a non-decreasing set, we would only obtain a lower bound for $d_r(\ev(\M))$. 
\end{rem}

If we wanted to compute $d_r(\ev(\M))$ directly using the definitions, usually this would require to consider $\mathcal{F}=\{ F=\{f_1,\dots,f_r\}\subset\mathcal{L}(\M ) \}$, check which sets of $\mathcal{F}$ are linearly independent, and compute $\Delta(I_\xu, F)$. Since 
$$
\abs{\mathcal{F}}=\binom{q^{s\abs{\M}}-1}{r},
$$
this becomes unfeasible in most cases. Using Theorem \ref{T:fpsharp}, we only need to check the value of $\Delta^*(N)$ for
$$
\abs{\mathcal{N}_r}=\binom{\abs{\M}}{r}
$$
sets of monomials. Moreover, the computation of $\Delta^*(N)$ itself is straightforward (e.g., using Lemma \ref{lemafootprint} and some variations of it).

In the next section, we show a way to find $d_r(\ev(\M))$ directly, without having to check the value of $\Delta^*(N)$ for every $N\in \mathcal{N}_r$, for some particular sets of monomials $\M$.

\section{Generalized Hamming weights of Reed-Muller-type codes over the norm-trace curve} \label{sec:GHW_RMlike}

Let $\left\{P_1,\ldots,P_n \right\} \subseteq \fqs^2$ be the rational points  of the extended norm-trace curve $\cX_u$ defined over the finite field $\F_{q^s}$ by the equation
\[
x^{u} = y^{q^{s - 1}} + y^{q^{s - 2}} + \cdots + y,
\]
where $n=u(q-1)q^{s-1}+q^{s-1}$ and $u$ is a positive integer such that $u \mid \frac{q^s - 1}{q - 1}$. Let $0\leq d\leq u(q-1)(q^{s-1}-1)$. The set of bivariate monomials of degree less than or equal to $d$ that are contained in $\Delta(y^{q^{s-1}},x^{u(q-1)+1})$ is denoted by
$$\M_{\leq d}:= \F_{q^s}[x,y]_{\leq d}\cap \Delta(y^{q^{s-1}},x^{u(q-1)+1}).$$
Let $\mathcal{L}(\Md)$ be the $\fqs$-subspace generated by $\Md$. The Reed-Muller-type code over the norm-trace curve is given by
\[ {\rm ev}(\Md) = \{ \left(f(P_1),\ldots,f(P_n)\right) : f \in \mathcal{L}(\Md) \}.\]

Let $\Cart_{u}(d)$ be an affine Cartesian code \cite{hiramaffinecartesian} obtained by evaluating $\mathcal{L}(\Md)$ at the points of a Cartesian set $Z_1\times Z_2 \subseteq \fqs \times \fqs$, where $\abs{Z_1}=u(q-1)+1$ and $\abs{Z_2}=q^{s-1}$. Note that the length and dimension of $\Cart_u(d)$ are the same as the Reed-Muller-type code over the norm-trace curve $\ev(\M_{\leq d})$. We start with a technical lemma that we will use later.

\begin{lem}\label{L:derecha}
Let $\M\subset \Delta(y^{q^{s-1}},x^{u(q-1)+1})$ be a decreasing set with $1\leq r\leq \abs{\M}$. Consider $N\subset \M$ with $\abs{N}=r$ such that 
$$
\abs{\Delta^*(N)}=\max\left\{ \abs{\Delta^*(M')}: M'\subset \M,\abs{M'}=r\right\},
$$
and let $a_1$ be the lowest power of $x$ appearing in any monomial in $N$. Then either $x^{a_1+u} \in N $ or $x^{a_1+u}\not \in \M$.
\end{lem}
\begin{proof}
We can assume $N=\{x^{a_i}y^{b_i}, 1\leq i\leq r\}\subset \M$, with $a_1\leq \cdots \leq a_r$. 
Assume $x^{a_1+u}\not \in N$ and $x^{a_1+u} \in \M$. By the definition of $\Delta^*(N)$, we have that $\abs{\Delta^*(N)}=\abs{\Delta^*(N\cup \{x^{a_1+u}\})}$. Due to the proof of Theorem~\ref{T:fpsharp}, we can find a set of polynomials $F\subset \mathcal{L}(\M)$ such that $F_{\ini}=N\cup \{x^{a_1+u}\}$ and 
$$
\abs{V_{\xu}(F)}=\abs{\Delta^*(N\cup \{x^{a_1+u}\})}=\abs{\Delta^*(N)}.
$$
This implies $d_{r+1}(\ev(\M))\leq \abs{\xu}-\abs{\Delta^*(N\cup \{x^{a_1+u}\})}=\abs{\xu}-\abs{\Delta^*(N)}=d_{r}(\ev(\M))$, where the last equality follows from the maximality of $\abs{\Delta^*(N)}$. This contradicts the monotonicity of the generalized Hamming weights from Theorem \ref{T:monotonicity}. 
\end{proof}

Let $M(r)$ denote the first $r$ elements of $\Md$ in descending lexicographic order with $x>y$. For the case with $\abs{Z_1}\leq \abs{Z_2}$, in \cite{BEELEN2018130} the authors prove that $M(r)$ is the set of monomials that maximizes $\abs{\Delta(y^{\abs{Z_2}},x^{\abs{Z_1}},M(r))}$. In the next result, we show that $M(r)$ also maximizes $\abs{\Delta^*(M(r))}$ if $u\neq \frac{q^s-1}{q-1}$. 

\begin{lem}\label{L:shadow}
Let $0\leq d\leq u(q-1)(q^{s-1}-1)$ and $N\subset \Md$ with $\abs{N}=r$. If $u\neq \frac{q^s-1}{q-1}$, then
$$
\abs{\Delta^*(M(r))}\geq \abs{\Delta^*(N)}.
$$
\end{lem}
\begin{proof}
Since $u\neq (q^s-1)/(q-1)$ and $u\mid (q^s-1)/(q-1)$, we have
$$
u\leq \frac{1}{2}(q^{s-1}+q^{s-2}+\cdots+1)<\frac{2q^{s-1}}{2}=q^{s-1}.
$$
Therefore, $u\leq q^{s-1}-1$. Let $N\subset \Md$ with $\abs{N}=r$ be such that 
$$
\abs{\Delta^*(N)}=\max\left\{ \abs{\Delta^*(M')}: M'\subset \Md,\abs{M'}=r\right\}.
$$
By Lemma \ref{L:derecha}, we can assume that either $x^{a_1+u}\in N$ or $x^{a_1+u}\not \in \Md$. If $x^{a_1+u}\not \in \Md$, then we have
$$
\abs{\Delta^*(M(r))}\geq \abs{\Delta^*(N)}.
$$
To see this, denote $v'=\min \{ u, u(q-1)+1-a_1 \}$ and note that 

\begin{equation}\label{eq:boxreduction}
\begin{aligned}
&\abs{\Delta^*(M(r))}=a_1(q^{s-1}-1)+\abs{\Delta\left( (M(r)/x^{a_1})\cup \{ y^{q^{s-1}},x^{v'}\}\right)},\\
&\abs{\Delta^*(N)}=a_1(q^{s-1}-1)+\abs{\Delta\left( (N/x^{a_1})\cup \{ y^{q^{s-1}},x^{v'}\}\right)},
\end{aligned}
\end{equation}
where $N/x^{a_1}$ denotes the set of the monomials of $N$ divided by $x^{a_1}$ (similarly for $M(r)/x^{a_1}$). Since $x^{a_1+u}\not \in \Md$, the $r$ elements of $M(r)/x^{a_1}$ and $N/x^{a_1}$ are contained in $\Delta(y^{q^{s-1}},x^{v'})$ (the power of $x$ of those monomials is lower than $v'$). By \cite[Prop. 3.8]{BEELEN2018130}, we know
$$
\abs{\Delta\left( (M(r)/x^{a_1})\cup \{ y^{q^{s-1}},x^{v'}\}\right)}\geq \abs{\Delta\left( (N/x^{a_1})\cup \{ y^{q^{s-1}},x^{v'}\}\right)}
$$
since this is the inequality that arises when considering a Cartesian code over a product of two sets of sizes $v'$ and $q^{s-1}$ (recall that $u\leq q^{s-1}-1$, which implies $v'\leq q^{s-1}-1$). 

If we assume now that $x^{a_1+u}\in N$, then (\ref{eq:boxreduction}) still holds, but some elements of $M(r)/x^{a_1}$ and $N/x^{a_1}$ are not contained in $\Delta(y^{q^{s-1}},x^{v'})$ because they are multiples of $x^{v'}$. Let 
$$
\begin{aligned}
\abs{(N/x^{a_1})\cap \Delta(y^{q^{s-1}},x^{v'})}=r'<r, \:
\abs{(M(r)/x^{a_1})\cap \Delta(y^{q^{s-1}},x^{v'})}=r''<r.
\end{aligned}
$$
By the definition of $M(r)$, we must have $r''\leq r'$. If $r''<r'$, there is $\tilde{r}$ such that $\abs{(M(r+\tilde{r})/x^{a_1})\cap \Delta(y^{q^{s-1}},x^{v'})}=r'$, and $(M(r+\tilde{r})/x^{a_1})\cap \Delta(y^{q^{s-1}},x^{v'})$ are precisely the first $r'$ elements of $\Delta(y^{q^{s-1}},x^{v'})$ in descending lexicographic order
Hence, by \cite[Prop. 3.8]{BEELEN2018130}, we have
$$
\abs{\Delta\left( (M(r+\tilde{r})/x^{a_1})\cup \{ y^{q^{s-1}},x^{v'}\}\right)}\geq \abs{\Delta\left( (N/x^{a_1})\cup \{ y^{q^{s-1}},x^{v'}\}\right)}.
$$
Since $M(r)/x^{a_1}\subset M(r+\tilde{r})/x^{a_1}$, we obtain the result. 
\end{proof}

Due to Lemma \ref{L:shadow}, we can directly obtain $d_r(\ev(\Md))$ when $u\neq \frac{q^s-1}{q-1}$, and we can also relate the generalized Hamming weights of these decreasing norm-trace codes with the ones from Cartesian codes. 

\begin{thm}\label{T:ghwRMlike}
Let $0\leq d\leq u(q-1)(q^{s-1}-1)$ and $1\leq r\leq \abs{\Md}$. If $u\neq \frac{q^s-1}{q-1}$, then
$r^{th}$ generalized Hamming weight is 
$$
d_r(\ev(\Md))=\abs{\xu}-\abs{\Delta^*(M(r))}=\abs{\xu}-\abs{\Delta(M(r),y^{q^{s-1}},x^{u(q-1)+1})}.
$$
As a consequence, if $u\leq \frac{q^{s-1}-1}{q-1}$, we have 
$$
d_r(\ev(\Md))=d_r(\Cart_u(d)).
$$
Moreover, for any $u$ such that $u\mid \frac{q^s-1}{q-1}$, we have the bound
$$
d_r(\ev(\Md))\geq d_r(\Cart_u(d)).
$$
\end{thm}
\begin{proof}
We start with the case $u\neq (q^s-1)/(q-1)$. By Theorem \ref{T:fpsharp} and Lemma \ref{L:shadow}, we have
$$
d_r(\ev(\Md))=\abs{\xu}-\abs{\Delta^*(M(r))}.
$$

Now we claim $\Delta^*(M(r))=\Delta(M(r),y^{q^{s-1}},x^{u(q-1)+1})$. Indeed, let $a_1$ be the smallest power of $x$ in $M(r)$. If $a_1+u\geq u(q-1)+1$, the claim follows from the definition of $\Delta^*$. If $a_1+u< u(q-1)+1$, then $x^{a_1}y^{b_1}\in M(r)$, for some $b_1$. If $b_1=0$, we have $x^{a_1}\mid x^{a_1+u}$. If $b_1>0$, by the definition of $M(r)$, this implies $x^{a_1+1}\in M(r)$. Since $x^{a_1+1}\mid x^{a_1+u}$, we have proved the claim. 
By \cite{BEELEN2018130}, if $u(q-1)\leq q^{s-1}-1$, we have $d_r(\Cart_u(d))=\abs{\xu}-\abs{\Delta(M(r),y^{q^{s-1}},x^{u(q-1)+1})}$. 

With respect to the bound, we note that for any set $N\subset \Md$, we have $\Delta^*(N) \subset  \Delta(N\cup \{y^{q^{s-1}},x^{u(q-1)+1}\})$. Thus, 
$$
d_r(\Cart_u(d))\leq \abs{\xu}-\abs{\Delta(N\cup \{y^{q^{s-1}},x^{u(q-1)+1}\})}\leq \abs{\xu} - \abs{\Delta^*(N)}.
$$
By choosing $N$ with maximum $\abs{\Delta^*(N)}$, we obtain the stated lower bound. 
\end{proof}

By Theorem \ref{T:ghwRMlike}, the codes $\ev(\M_{\leq d})$ are on par with $\Cart_u(d)$ if $u\leq \frac{q^{s-1}-1}{q-1}$, and they can perform better in other cases. Now, we give an example with each of the cases of Theorem \ref{T:ghwRMlike}, showing how decreasing norm-trace codes can outperform Cartesian codes.

\begin{ex}\label{ex:ghws}
Let $q=3$, $s=2$, $d=4$ and $r=3$. For $u=1$, we have $u\leq (q^{s-1}-1)/(q-1)=1$, and
$$
d_3(\ev(\M_{\leq 4}))=d_3(\Cart_1(4))=3.
$$
This is obtained with the set $M(3)=\{y^2x^2,yx^2,x^2\}$. If we consider $u=2$ instead, we obtain
$$
d_3(\ev(\M_{\leq 4}))=6>5=d_3(\Cart_2(4)).
$$
In this case, $M(3)=\{x^4,yx^3,x^3\}$. Finally, if $u=4=(q^s-1)/(q-1)$, we have
$$
d_3(\ev(\M_{\leq 4}))=17>9=d_3(\Cart_4(4)).
$$
This value is attained with $N=\{yx^3,yx^2,y^2x^2\}\neq M(3)=\{x^4,yx^3,x^3\}$ ($\abs{\xu}-\Delta^*(M(3))$ gives 18 in this case).

Moreover, these examples also showcase the importance of considering the extra monomial $x^{a+u}$ in the definition of $\Delta^*(F)$. Otherwise, the footprint approach would give the same bound for the generalized Hamming weights as the one we have for Cartesian codes.
\end{ex}

In Example \ref{ex:ghws}, for the case $u=4$, the set of monomials $N$ which gives the maximum value of $\Delta^*(N)$ does not correspond to $M(r)$, nor to the analogous set obtained by choosing $y<x$ in the ordering of the elements. This shows that for $u=\frac{q^s-1}{q-1}$, the description of the set $N$ that maximizes $\abs{\Delta^*(N)}$ gets more involved. We will now study the case $u=\frac{q^s-1}{q-1}$ in more detail.

Let $0\leq a_1\leq u(q-1)$, and let 
$$
\mathcal{M}^{a_1}_{\leq d}:=\fqs[x,y]_{\leq d-a_1}\cap \Delta(y^{q^{s-1}},x^u).
$$
Consider $M_{a_1}'(r)$, the set formed by the first $r$ elements of $\mathcal{M}^{a_1}_{\leq d}$ in descending lexicographic order with $y<x$ (for $r\leq 0$, we will define $M_{a_1}'(r)=\emptyset$). Since $u=\frac{q^s-1}{q-1}>q^{s-1}$, by \cite{BEELEN2018130} this set maximizes $\abs{\Delta(y^{q^{s-1}},x^u,N)}$ for any $N\subset \mathcal{M}^{a_1}_{\leq d}$. If $d\geq a_1+u$, we consider the set
$$
T_{a_1}:=\{ x^ay^b\mid a_1+u\leq a \leq d, 0\leq b \leq d-a \}.
$$
We denote $r_{a_1}:=\abs{T_{a_1}}=\frac{(d-a_1-u)(d-a_1-u+1)}{2}$. In the next result, we show that we only need to consider a few sets of monomials to compute the generalized Hamming weights as in Theorem~\ref{T:fpsharp}.

\begin{prop}\label{P:specialcase}
Let $u=\frac{q^s-1}{q-1}$, $1\leq d \leq u(q-1)(q^{s-1}-1)$ and $1\leq r\leq \abs{\Md}$. For each $0\leq a_1\leq d$, consider
$$
\mathcal{N}(a_1):=\{ N \subset \Md \; :  \abs{N}=r, \; a_1=\min \{ a\mid x^ay^b\in N, \text{ for some }b \} \}.
$$
Denote $r'=r_{a_1}+d-q^{s-1}-a_1+2$ and consider
$$
N_{a_1}=\begin{cases}
    x^{a_1}\cdot M'_{a_1}(r) &\text{ if } d<a_1+u, \\
    x^{a_1}\cdot M'_{a_1}(r-r_{a_1})\cup T_{a_1}  &\text{ if } d\geq a_1+u \text{ and } r\geq r'.
\end{cases}
$$
Then we have
$$
\abs{N_{a_1}}=\max \left \{ \abs{\Delta^*(N)} : N\subset \mathcal{N}(a_1), \abs{N}=r \right\}.
$$
Moreover, 
$$
d_r(\ev(\Md))=\abs{\xu}-\max\biggl\{ \abs{\Delta^*(N_{a_1})} \biggr.: 
\begin{aligned}
&0\leq a_1\leq \min\{d,u(q-1)\} \\
& d<a_1+u \text{ or } d\geq a_1+u \text{ and } r\geq r'.
\end{aligned}\biggl.\biggr\}.
$$
\end{prop}
\begin{proof}
Fix $0\leq a_1\leq d$ and $1\leq d \leq u(q-1)(q^{s-1}-1)$. Assume first $d<a_1+u$. For any $N\in \mathcal{N}(a_1)$, we have the monomials $x^ay^b$ with $0\leq a <a_1$, and the monomials $x^ay^b$ with $a_1+u\leq a \leq u(q-1)$ in $\Delta(y^{q^{s-1}},x^{u(q-1)+1}))\setminus \Delta^*(N)$. Therefore, the differences between different sets $N\in \mathcal{N}(a_1)$ arise from the monomials $x^ay^b$ with $a_1\leq a <a_1+u$, and we have to find the set $N$ that maximizes the number of monomials in $\Delta^*(N)$ in this region. This is equivalent to finding the set $N' \subset \Md^{a_1}$ that maximizes the number of monomials in $\Delta(y^{q^{s-1}},x^u,N)$ (we then multiply the monomials of $N'$ by $x^{a_1}$). By \cite{BEELEN2018130}, we know that this set is $M_{a_1}'(r)$. 

If $d\geq a_1+u$, we have $x^{a_1+u}\in \Md$. Following the idea of the proof of Lemma \ref{L:derecha}, if $r\geq r'>r_{a_1}$, then $T_{a_1}\subset N_{a_1}$. Restricting ourselves to the subsets $N\subset \Md$ such that $T_{a_1}\subset N$, the problem of finding $N_{a_1}$ among them is as in the previous case, since again the differences arise from the monomials $x^ay^b$ with $a_1\leq a <a_1+u$. Note that $r\geq r'$ ensures that $\min\{a\mid x^ay^b\in x^{a_1}\cdot M'_{a_1}(r-r_{a_1})\}=a_1$.

We will use Theorem~\ref{T:fpsharp} to finish the proof. It is clear that 
$$
\max\left\{\abs{\Delta^*(N)}, N \subset \mathcal{N}_r\right\}=\max\left\{\abs{\Delta^*(N_{a_1})}, 0\leq a_1\leq \min\{d,u(q-1)\}\right\}.
$$
Assume $r< r'$. By Lemma \ref{L:derecha}, if 
\begin{equation}\label{eq:hypmax}
\abs{\Delta^*(N_{a_1})}=\max\left\{\abs{\Delta^*(N)}, N \subset \mathcal{N}_r\right\},
\end{equation}
then we should have $T_{a_1}\subset N_{a_1}$. This is a contradiction if $r<r_{a_1}$. If $r\geq r_{a_1}$, we have $T_{a_1}\subset N_{a_1}$, and the argument from above shows that the set $N$ with $r-r_{a_1}$ elements that maximizes the number of monomials in $\Delta(y^{q^{s-1}},x^u,N)$ is $M_{a_1}'(r-r_{a_1})$. The set $x^{a_1}\cdot M'_{a_1}(r-r_{a_1})\cup T_{a_1}$ has $r$ monomials. We have
$$
a'_1=\min\{a\mid x^ay^b\in x^{a_1}\cdot M'_{a_1}(r-r_{a_1})\cup T_{a_1}\}=\begin{cases}
    a_1+u &\text{ if } r=r_{a_1},\\
    d-q^{s-1}-r+r_{a_1}+2 &\text{ if } r>r_{a_1}.
\end{cases}
$$
Since $r<r'$, this is not equal to $a_1$. This means that 
$$
\abs{\Delta^*(N_{a'_1})}> \abs{\Delta^*(N_{a_1})},
$$
with $a'_1>a_1$, a contradiction with Equation~(\ref{eq:hypmax}). 
\end{proof}

\begin{ex}
Continuing with Example \ref{ex:ghws}, for the case $u=4$ it is easy to check that
$$
N=\{yx^3,yx^2,y^2x^2\}=x^2\cdot \{ yx,y,y^2\}=x^2 \cdot M_2'(3)=N_2.
$$
\end{ex}

Theorem \ref{P:specialcase} allows us to compute the generalized Hamming weights by considering, at most, $\min\{d,u(q-1)\}+1$ sets of monomials. Although this is already relatively easy to compute, we show now a particular case in which we can directly obtain the generalized Hamming weights.

\begin{cor}
Let $u=\frac{q^s-1}{q-1}$, $1\leq d \leq q^{s-1}$ and $1\leq r\leq \abs{\Md}$. Then 
$$
d_r(\ev(\Md))=\abs{\xu}-\abs{\Delta^*(M'_0(r))}.
$$
\end{cor}
\begin{proof}
From Proposition \ref{P:specialcase}, if $\mathcal{N}_{a_1}\neq \emptyset$ we have
$$
N_{a_1}=x^{a_1}\cdot M'_{a_1}(r).
$$
Consider $b_1=\min \{b\mid y^b \in M'_{a_1}(r)\}$, $b_2=\min \{b \mid x^ay^b\in M'_{a_1}(r) \text{ for some }a>0\}$, and $\lambda=\min\{a\mid x^ay^{b_1}\in M'_{a_1}(r)\}$. Defining $a_2=a_1+\lambda$ and using Lemma \ref{lemafootprint}, it is straightforward to check that
$$
\abs{\Delta^*(N_{a_1})}=a_1q^{s-1}+b_2u+(a_2-a_1)(b_1-b_2).
$$
If we want to compare this with $\abs{\Delta^*(N_{a'_1})}$, for some $a'_1=a_1+\varepsilon$ with $\varepsilon>0$, we see that the corresponding values of $b_1,b_2,a_2$ for $a'_1$ become $b'_1=b_1-\varepsilon$, $b'_2=b_2-\varepsilon$, and $a_2=a_2+\varepsilon$. The difference is
$$
\abs{\Delta^*(N_{a_1})}-\abs{\Delta^*(N_{a'_1})}=-\varepsilon q^{s-1}+\varepsilon u=\epsilon\left(  \frac{q^s-1}{q-1}-q^{s-1}\right)>0,
$$
which means that $N_0=M_0'(r)$ maximizes $\abs{\Delta^*(N_{a_1})}$. 
\end{proof}

\section{Relative generalized Hamming weights of decreasing norm-trace codes and quantum codes}\label{sec:RGHW}

Relative generalized Hamming weights were introduced to characterize the security of linear ramp secret-sharing schemes. The $r$-the relative generalized Hamming weight of a pair of codes $C_1$ and $C_2$ with $C_2\subset C_1$ is the smallest number such that any set of shares of that size allows to recover $r$ $q$-bits of information about the secret for any set of that size, and the $m$-th relative generalized Hamming weight of the pair of codes $C_1^\perp \subset C_2^\perp$ gives the smallest size of a set of shares that can determine $m$ $q$-bits of information about the secret \cite{matsumotoRGHW,geilrghws}. These parameters have been studied for many different families of codes \cite{BEELEN2018130, galindonested, geilRGHWReedMuller}, and in particular, for one-point algebraic geometry codes \cite{geilrghws,christensennested}. In what follows, we study the relative generalized Hamming weights of decreasing norm-trace codes and their relation to impure quantum codes. We start by defining the relative generalized Hamming weights of a pair of codes using the definition from \cite[Lemma 1]{luoRGHW}.

\begin{defn} \label{d:rghw}
Let $C'\subset C\subset \fqs^n$ be two linear codes, and consider $r$ with $1\leq r \leq \dim C -\dim C'$. The $r$-th \textit{relative generalized Hamming weight} of $C$ and $C'$, denoted $M_r(C,C')$, is
$$
M_r(C,C')=\min \{ \abs{\supp(D)} : D \text{ is a subcode of $C$ with } \dim D=r, \; D\cap C'=\{0\}\}.
$$
It is useful to note that the condition $D\cap C'=\{0\}$ is equivalent to the condition $D\setminus\{0\}\subset C\setminus C'$. That is why, in many cases, one studies the subcodes of $C\setminus C'$ (disregarding the vector $0$) for computing the relative generalized Hamming weights.
\end{defn}

From the definition, it is clear that the relative generalized Hamming weights are a generalization of generalized Hamming weights obtained by considering $C'=\{0\}$. Moreover, for $1\leq r \leq \dim C-\dim C'$, we have
\begin{equation}\label{eq:boundRGHWwithGHW}
M_r(C,C')\geq d_r(C).
\end{equation}
Thus, we can always bound the relative generalized Hamming weights from below using the generalized Hamming weights, and we can use the results from the previous sections to obtain bounds for the relative generalized Hamming weights. 
With respect to decreasing norm-trace codes, it is possible to use the techniques from the previous section to obtain the relative generalized Hamming weights. The idea is similar to that of the relative footprint from \cite{hiramRGHW}. 

\begin{thm}\label{T:RGHW}
Let $\M_2\subset\M_1\subset  \Delta(y^{q^{s-1}},x^{u(q-1)+1})$ be decreasing sets. Let $1\leq r \leq \abs{\M_1}-\abs{\M_2}$, and let $\M_{1,2}^r$ be the set of all subsets $M$ of $\{\ini(f)\mid f\in \mathcal{L}(\M_1)\setminus  \mathcal{L}(\M_2)\}$ with $r$ elements. Then
$$
M_r(\ev(\M_1),\ev(\M_2))\geq \abs{\xu}-\max \{\Delta^*(M)\mid M \in\M_{1,2}^r\}.
$$
Moreover, if $x^ay^b\in \M_1\setminus \M_2$ implies $x^ay^b>x^{a'}y^{b'}$, for all $x^{a'}y^{b'}\in \M_2$ (we abbreviate this by writing $\M_1\setminus \M_2> \M_2$), then 
\begin{equation}\label{eq:drel}
M_r(\ev(\M_1),\ev(\M_2))= \abs{\xu}-\max \{\Delta^*(M)\mid M \subset\M_1\setminus \M_2,\;\abs{M}=r\}.
\end{equation}

\end{thm}
\begin{proof}
Consider $R_r$ the set of all subsets $F=\{f_1,\dots,f_r\}\subset \mathcal{L}(\M_1)\setminus  \mathcal{L}(\M_2)$ of size $r$ such that $\ini(f_1),\dots,\ini(f_r)$ are different. From \cite{hiramRGHW} we have
\begin{equation*}\label{eq:RGHWAC}
M_r(\ev(\M_1),\ev(\M_2))=\abs{\xu}-\max \{\abs{\Delta(I_\xu,F)}\mid F\in R_r\}.
\end{equation*}   
Taking into account that $\abs{\Delta(I_\xu,F)}=\abs{V_\xu(F)}$ and the bound (\ref{lowerbound}), we obtain
$$
\begin{aligned}
M_r(\ev(\M_1),\ev(\M_2))&\geq \abs{\xu}-\max \{\Delta^*(F)\mid F\in R_r\}\\
&=\abs{\xu}-\max \{\Delta^*(M)\mid M \in\M_{1,2}^r\}.
\end{aligned}
$$
If $\M_1\setminus \M_2>\M_2$, then, for any $f\in \mathcal{L}(M_1)\setminus \mathcal{L}(M_2)$, we have $\ini(f)\in \M_1\setminus \M_2$. Indeed, if we had $\ini(f)\in \M_2$, by the hypothesis $\M_1\setminus \M_2>\M_2$ this would imply that all the monomials of $f$ are in $\M_2$, which implies $f\in \mathcal{L}(M_2)$, a contradiction. Thus, $\M_{1,2}^r$ is the set of all subsets $M$ of monomials in $\M_1\setminus \M_2$ of size $r$, and we have to prove that the equality holds in the bound we have proved above. 

By the proof of Theorem \ref{T:fpsharp}, for each set $ M \subset\M_1\setminus \M_2$ with size $r$, we can find a set of $r$ polynomials $F$ such that $M=F_\ini$. Since $F_\ini\subset \M_1\setminus \M_2$, this implies that $F\subset \mathcal{L}(M_1)\setminus \mathcal{L}(M_2)$. Moreover, we have 
$$
\abs{V_\xu(F)}=\Delta^*(F)=\Delta^*(M).
$$
By choosing $M$ such that 
$$
\max \{\Delta^*(M): M\subset\M_1\setminus \M_2,\;\abs{M}=r\}=\Delta^*(M),
$$
we see that the lower bound we have obtained is attained by $F\subset \mathcal{L}(\M_1)\setminus \mathcal{L}(\M_2)$, which concludes the proof.
\end{proof}

From \cite{decreasingnormtrace}, we have the following result.

\begin{thm}\label{T:dual}
Assume $\xu=\{P_1,\dots,P_n\}$ and let $\M\subset \Delta(y^{q^{s-1}},u^{u(q-1)+1})$ be a decreasing set. The dual code $\ev(\M)^\perp$ is equivalent to $\ev(\M^c)$, where
$$
\M^c:=\left \{ \frac{x^{u(q-1)}y^{q^{s-1}-1}}{x^ay^b}: x^ay^b\in \Delta(y^{q^{s-1}},x^{u(q-1)+1})\setminus \M\right\}.
$$
More precisely,
$$
\ev(\M)^\perp=\beta \star \ev(\M^c),
$$
where
$$
\beta_i:=\begin{cases}
    u^{-1} &\text{ if the } x\text{-coordinate of }P_i \text{ is nonzero,}\\
    1 &\text{ otherwise,}
\end{cases}
$$
and $\star$ denotes the component-wise product of vectors.
\end{thm}

\begin{rem}
Since $\M^c$ is decreasing if $\M$ is decreasing, and $\beta$ does not depend on $\M$, we can also use Theorem \ref{T:RGHW} to compute the relative generalized Hamming weights of the duals of a pair of decreasing norm-trace codes.
\end{rem}

As stated in \cite{decreasingnormtrace}, the family of decreasing norm-trace codes contains one-point algebraic geometry codes over the norm-trace (in particular, Hermitian codes). With the usual notation $C(D_u,\lambda P_\infty)$ for one-point algebraic geometry codes, we consider $D_u=\sum_{P_i\in \xu}P_i$. Then, the evaluation of
$$
\mathcal{L}_\lambda=\left\{ x^ay^b\in \Delta(y^{q^{s-1}},x^{u(q-1)+1})\mid aq^{s-1}+bu\leq \lambda \right\}
$$
gives a basis for $C(D_u,\lambda P_\infty)$. Since this is a decreasing set, we can use Theorem \ref{T:fpsharp} to compute the generalized Hamming weights of this code. This recovers, in particular, the result by Barbero and Munuera \cite{Barbero_Munuera_00}, although in a less explicit way. Nevertheless, our techniques cover all the one-point algebraic geometry codes over the extended norm-trace curve.

As a particular case of Theorem \ref{T:RGHW}, if one considers a pair of one-point algebraic geometry codes over $\xu$, then we always have the equality in (\ref{eq:drel}). The relative generalized Hamming weights of one-point algebraic geometry codes have been previously studied in \cite{geilrghws} using the Feng-Rao bound. In our case, we use a footprint-like approach. Moreover, the authors in \cite{geilrghws} check that their bound is often tight for Hermitian codes, and in our case, we prove that our bound is always tight for norm-trace codes (in particular, for Hermitian codes). The particular case of Hermitian codes is studied in more depth in~\cite{christensennested}. 

The first relative generalized Hamming weight (also called relative minimum distance) of a pair of codes (and their duals) also has applications to quantum codes. As we show next, from a nested pair of codes $C_2\subset C_1$, one can construct a quantum error-correcting code (QECC) that encodes $\dim C_1-\dim C_2$ logical qudits using $n$ physical qudits. QECCs codes can correct two different types of errors, namely phase-shift errors and qudit-flip errors. The error-correction capabilities of the code for each type of error are denoted by $\delta_z$ and $\delta_x$, respectively, meaning that this code can correct up to $\lfloor (\delta_z-1)/2\rfloor$ phase-shift errors and up to $\lfloor (\delta_x-1)/2\rfloor$ qudit-flip errors. The error correction capabilities of the QECC are given precisely by the first relative generalized Hamming weight of the nested pair of codes $C_2\subset C_1$ and their duals $C_1^\perp \subset C_2^\perp$. We now provide the precise statement of the CSS construction that we use to construct quantum codes, which was obtained independently by Calderbank and Shor \cite{calderbankgoodquantum} and Steane \cite{Steane_96}. 

\begin{thm}\label{T:css}
Let $C_2\subset C_1\subset \fqs^n$ be linear codes. Then, we can construct an asymmetric quantum code with parameters 
$$
[[n,\dim C_1-\dim C_2,\delta_z/\delta_x]]_{q^s},
$$
where $\delta_z=M_1(C_1,C_2)$ and $\delta_x=M_1(C_2^\perp,C_1^\perp)$. 
\end{thm}
If $M_1(C_1,C_2)=d_1(C_1)$ and $M_1(C_2^\perp,C_1^\perp)=d_1(C_2^\perp)$, we say that the corresponding quantum code is \textit{impure}, and is called \textit{pure} otherwise (recall (\ref{eq:boundRGHWwithGHW})). Impure quantum codes are desired since they have decoding advantages and are, in general, difficult to obtain using classical codes~\cite{duadicImpure,grasslImpure}.

Notice that, by considering $C_1^\perp \subset C_2^\perp$ in the previous result, we obtain a quantum code with the same parameters, but with $\delta_z$ and $\delta_x$ exchanged. In all of our examples, we will choose the codes with $\delta_z\geq \delta_x$ since phase-shift errors are more likely to occur than qudit-flip errors \cite{ioffe}, but one can always obtain also a quantum code with the minimum distances exchanged by using the duals.

The following example shows how to obtain asymmetric quantum codes with good parameters using one-point codes over the norm-trace curve. 

\begin{ex}
Let $q=5$, $s=2$, and $u=3$. Then, we have $\abs{\xu}=65$. Consider the set $\M_2=\{1,y,x,y^2\}$ and $\M_1=\M_2\cup \{xy\}$. The codes $C_1=\ev(\M_1)$ and $C_2=\ev(\M_2)$ have parameters $[65,5,57]_{25}$ and $[65,4,59]_{25}$, respectively. The minimum distances can be obtained with Theorem \ref{T:fpsharp} (or \cite[Thm. 4.3]{decreasingnormtrace}). With respect to the dual codes, we use Theorem \ref{T:dual}. Note that $\M_1^c=\Delta(x^{13},y^5)\setminus\{ (x^{12}y^3)/M \mid M\in \M_1\}=\Delta(x^{13},y^5)\setminus\{x^{11}y^3,x^{12}y^2,x^{11}y^4,x^{12}y^3,x^{12}y^4\}$, and $\M_2^c=\M_1^c\cup \{  x^{11}y^3\}$. Then $C_1^\perp=\beta\star \ev(\M_1^c)$ and $C_2^\perp=\beta \star\ev(\M_2^c)$. As before, we can compute the parameters and obtain $[65,60,3]_{16}$ and $[65,61,3]_{16}$. If we just use the minimum distances of $C_1$, $C_2$, and their duals, using Theorem \ref{T:css} we obtain a quantum code with parameters $[[65,1,\geq 57/\geq 3]]_{16}$. In this case, we can use (\ref{eq:drel}) since $C_1$ and $C_2$ are constructed with monomials with increasing weighted degrees. Note that, since $d_1(C_1)<d_2(C_2)$, we have $d_1(C_1)=M_1(C_1,C_2)=57$. However, for the duals, we have $M_1(C_2^\perp,C_1^\perp)=4>d_1(C_1^\perp)=3$. This is because $\M_2^c\setminus \M_1^c=\{ x^{11}y^3\}$, and it is straightforward to check that $\Delta^*(x^{11}y^3)=61$. Thus, by using Theorem \ref{T:css}, this time with the relative minimum distances, the corresponding quantum code has parameters $[[65,1,57/4]]_{16}$, and this code is impure.
\end{ex}

We now introduce a way to check the goodness of these codes that have been previously used in \cite{galindonested,christensennested}. Assume that we have a pair of linear codes $C_2\subset C_1\subset \fqs^n$ of codimension $\ell$ and $M_1(C_1,C_2)=\delta_z$. Then we consider the greatest value $g(\ell,\delta)$ such that the tables of best-known linear codes ensure the existence of $C,C'\subset \fqs^n$ with $\dim C-\dim C'=\ell$, $d_1(C)=\delta_z$, and $d_1({C'}^\perp)\geq g(\ell,\delta_z)$. If we had $C'\subset C$, by Theorem \ref{T:css} this would give rise to a quantum code with parameters $[[n,\ell, \geq \delta_z/\geq g(\ell,\delta_z)]]_{q^s}$. In Table \ref{T:table}, we show some of the quantum codes that we can obtain by considering $C_1=C(D_u,\lambda_1 P_\infty)$ and $C_2=C(D_u,\lambda_2 P_\infty)$, with $\lambda_2<\lambda_1$, in Theorem \ref{T:css}. Note that the codes we obtain can give the same parameters as the best-known codes and even outperform them in some cases. This is remarkable since, using the best-known codes, we do not have $C'\subset C$ in general (in fact, it is quite unlikely). The most commonly used table of best-known linear codes is \cite{codetables}. However, this table only covers linear codes over finite field sizes of at most 9. Since the finite fields we are considering for norm-trace codes can be bigger, we will use \cite{mint} instead to compute $g(\ell,\delta)$. Also, since a similar approach is taken in \cite{christensennested} for Hermitian codes, we will not consider the case with $s=2$ and $u=(q^s-1)/(q-1)=q+1$, since it corresponds precisely to the Hermitian curve.

In Table \ref{T:table}, if there is no code $C$ with $d_1(C)=\delta_z$, we have considered the code with the highest dimension and $d_1(C)> \delta_z$, and we have marked this cases with $^*$. Impure codes are marked with $^\star$, and all of them satisfy $M_1(C_1,C_2)=d_1(C_1)$, $M_2(C_2^\perp,C_1^\perp)>d_2(C_2^\perp)$, except the codes $[[32,1,15/8]]_8$ and $[[32,1,12/10]]_8$, which have both $M_1(C_1,C_2)>d_1(C_1)$ and $M_2(C_2^\perp,C_1^\perp)>d_2(C_2^\perp)$. We highlight a couple of particular cases now. For $q=5$, $s=2$, and $u=3$, we obtain several codes whose $d_x$ surpasses the value $g(\ell,\delta_z)$. However, these codes are not impure, and what is happening is that we are actually using codes that surpass the best-known parameters of linear codes in \cite{mint}. Indeed, the codes $[65,3,60]_{25}$, $[65,4,59]_{25}$ and $[65,55,8]_{25}$, corresponding to $\lambda$ equal to 5, 6 and 58, respectively, have better minimum distance than the codes stated in \cite{mint} for the same length and dimension (this explains why $g(2,60)$ and $g(3,59)$ cannot be computed for $[[65,2,60/2]]_{25}$ and $[[65,3,59/2]]_{25}$, respectively). We also highlight the code $[[45,1,41/4]]_{25}$, for $q=5$, $s=2$ and $u=2$. This code has $\delta_x=4>3=g(\ell,\delta_z)$, but the classical codes used for its construction do not have better parameters than those in \cite{mint}. Therefore, the fact that the code performs better than the theoretical quantum code constructed with the best-known linear codes shows how, by considering the relative minimum distance in Theorem \ref{T:css}, we can obtain impure quantum codes with excellent performance. We also note that all the codes presented in Table \ref{T:table} surpass the Gilbert-Varshamov bound for asymmetric quantum codes \cite[Thm. 4]{matsumotoAsymmetricGV}. 

\begin{table}[ht]
\centering
\begin{tabular}{c|c|c}
 \hline 
 $(\lambda_1,\lambda_2)$&Parameters&$g(\ell,\delta_z)$ \\
  \hline \hline
  &$q=2, s=3, u=7$\\ 
  \hline
$(4,0)$& $[[32,1,28/2]]_8$& 2  \\

$(8,7)$& $[[32,1,24/3]]_8$& 4  \\
$(11,8)$& $[[32,1,21/4]]_8^\star$& $4^*$  \\
$(15,14)$& $[[32,1,18/6]]_8^\star$& 6  \\
$(19,18)$& $[[32,1,15/8]]_8^\star$& 8  \\
$(23,22)$& $[[32,1,12/10]]_8^\star$& 11  \\

$(7,0)$& $[[32,2,25/2]]_8$& 2  \\
$(12,8)$& $[[32,2,20/4]]_8^\star$& 5  \\
$(16,14)$& $[[32,2,16/5]]_8^\star$& 6  \\

$(8,0)$& $[[32,3,24/2]]_8$& 2  \\
\hline \hline
  &$q=2, s=4, u=15$\\ 
  \hline
$(8,0)$& $[[128,1,120/2]]_{16}$& 2  \\
$(16,15)$& $[[128,1,112/3]]_{16}$& 3  \\
$(23,16)$& $[[128,1,105/4]]_{16}^\star$& $4^*$  \\
$(31,30)$& $[[128,1,98/6]]_{16}^\star$& $7$  \\

$(24,16)$& $[[128,2,104/4]]_{16}^\star$& 4  \\
$(32,30)$& $[[128,2,96/5]]_{16}^\star$& 7  \\

\hline \hline
  &$q=3, s=2, u=2$\\ 
  \hline
$(2,0)$ & $[[15,1,13/2]_{9}$ & 2 \\
$(4,3)$ & $[[15,1,11/3]_{9}$ &  3\\
$(5,4)$ & $[[15,1,10/4]_{9}$ &  4\\
$(6,5)$ & $[[15,1,9/5]_{9}$ &  5\\
$(7,6)$ & $[[15,1,8/6]_{9}$ &  6\\
$(8,7)$ & $[[15,1,7,7]_{9}$ &  7\\

$(3,0)$ & $[[15,2,12/2]_{9}$ & 2 \\
$(5,3)$ & $[[15,2,10/3]_{9}$ & 3 \\
$(6,4)$ & $[[15,2,9/4]_{9}$ &  4\\
$(7,5)$ & $[[15,2,8/5]_{9}$ & 5 \\
$(8,6)$ & $[[15,2,7/6]_{9}$ & 6 \\

$(4,0)$ & $[[15,3,11/2]_{9}$ & 2 \\
$(6,3)$ & $[[15,3,9/3]_{9}$ & 3 \\
$(7,4)$ & $[[15,3,8/4]_{9}$ & 4 \\
$(8,5)$ & $[[15,3,7/5]_{9}$ &  5\\
$(9,6)$ & $[[15,3,6/6]_{9}$ & 6 \\
\end{tabular}
\begin{tabular}{c|c|c}
 \hline 
 $(\lambda_1,\lambda_2)$&Parameters&$g(\ell,\delta_z)$ \\
  \hline \hline
  &$q=5, s=2, u=2$\\ 
  \hline
$(2,0)$& $[[45,1,43/2]]_{25}$& 2  \\
$(4,2)$& $[[45,1,41/3]]_{25}^\star$& 2  \\
$(6,5)$& $[[45,1,39/4]]_{25}$& 4  \\
$(8,7)$& $[[45,1,37/5]]_{25}$& 5  \\
$(9,8)$& $[[45,1,36/6]]_{25}$& 6  \\
$(10,9)$& $[[45,1,35/7]]_{25}$& 7  \\

$(4,0)$& $[[45,2,41/2]]_{25}$& 2  \\
$(7,5)$& $[[45,2,38/4]]_{25}$& 4  \\
$(8,6)$& $[[45,2,36/5]]_{25}$& 5  \\
$(9,7)$& $[[45,2,35/6]]_{25}$& 6  \\

$(5,0)$& $[[45,3,40/2]]_{25}$& 2  \\
$(8,5)$& $[[45,3,37/4]]_{25}$& 4  \\
$(10,7)$& $[[45,3,35/5]]_{25}$& 5  \\

  \hline \hline
  &$q=5, s=2, u=3$\\ 
  \hline
$(3,0)$& $[[65,1,62/2]]_{25}$& $2^*$  \\

$(6,5)$& $[[65,1,59/3]]_{25}$& 2  \\
$(8,6)$& $[[65,1,57/4]]_{25}^\star$& 4  \\
$(11,10)$& $[[65,1,54/6]]_{25}^\star$& 6  \\
$(14,13)$& $[[65,1,51/8]]_{25}$& 7  \\
$(16,15)$& $[[65,1,49/9]]_{25}$& 9  \\

$(5,0)$& $[[65,2,60/2]]_{25}$& $-$  \\
$(8,5)$& $[[65,2,57/3]]_{25}$& 3  \\
$(9,6)$& $[[65,2,56/4]]_{25}^\star$& 4  \\
$(12,10)$& $[[65,2,53/5]]_{25}$& 6  \\
$(14,12)$& $[[65,2,51/6]]_{25}$& 6  \\
$(15,13)$& $[[65,2,50/8]]_{25}$& 7  \\

$(6,0)$& $[[65,3,59/2]]_{25}$& $-$  \\
$(9,5)$& $[[65,3,56/3]]_{25}$& 3  \\
$(13,10)$& $[[65,3,52/5]]_{25}$& 6  \\
$(15,12)$& $[[65,3,50/6]]_{25}$& 6  \\
$(16,13)$& $[[65,3,49/8]]_{25}$& 7  \\

\end{tabular}
\caption{Parameters of the quantum codes corresponding to the pairs of codes $C(D_u,\lambda_1 P_\infty)$, $C(D_u,\lambda_2 P_\infty)$, with $\lambda_2<\lambda_2$. The star $^\star$ means the code is impure, and the asterisk $^*$ means that there is no code $C$ with $d_1(C)=\delta_z$ in \cite{mint}.}\label{T:table}
\end{table}

\section{Future work}  
Similarly to Section \ref{sec:GHW_RMlike}, it would be interesting to study how to construct the set of monomials that gives the maximum value of $\Delta^*$ for one-point algebraic geometry codes over the extended norm-trace curve. Similarly, another research direction would be to generalize these techniques for other families of curves and, more generally, different families of decreasing evaluation codes.

\bibliographystyle{abbrv}
\bibliography{BIBR}

\end{document}